\documentclass[final,copyright,creativecommons]{eptcs}
\providecommand{\event}{LSFA 2012}

\usepackage{soul}
\usepackage{wrapfig}
\usepackage{cite}
\usepackage{multicol}
\usepackage[all]{xy}
\usepackage{prooftree}
\usepackage{amssymb}
\usepackage{amsmath}
\usepackage{amsfonts}
\usepackage{amsthm}
\usepackage{float}
\usepackage{stmaryrd}
\floatstyle{boxed}\restylefloat{figure}
\hypersetup{
  colorlinks=true,
  citecolor=blue,
  linkcolor=black,
  urlcolor=black,
  pdftitle={Non determinism through type isomorphism},
  pdfauthor={Alejandro D\'iaz-Caro and Gilles Dowek}
}

\newcommand{\OurCalculus}{\ensuremath{\lambda_+}}
\newcommand{\eq}{\ensuremath{\rightleftarrows}}
\newcommand{\re}{\ensuremath{\hookrightarrow}}
\newcommand{\subst}[2]{\ensuremath{[{#1}/{#2}]}}
\newcommand{\tf}{\mbox{\bf T\!F}}
\newcommand{\true}{\mbox{\bf T}}
\newcommand{\false}{\mbox{\bf F}}
\newcommand{\ve}[1]{\ensuremath{\mathrm{\textbf{#1}}}}
\newcommand{\cond}[1]{\ensuremath{\mbox{\small$[#1]\,$}}}

\newtheorem{theorem}{Theorem}[section]
\newtheorem{lemma}[theorem]{Lemma}

\theoremstyle{definition}
\newtheorem{example}[theorem]{Example}

\title{Non determinism through type isomorphism}
\author{Alejandro D\'iaz-Caro\thanks{This work was supported by grants from DIGITEO and R\'egion \^Ile-de-France.}
\institute{Universit\'e Paris 13, Sorbonne Paris Cit\'e, LIPN\\
Universit\'e Paris-Ouest Nanterre La D\'efense\\
INRIA}
\and Gilles Dowek
\institute{INRIA\\
23 avenue d'Italie, CS 81321,\\ 75214 Paris Cedex 13}
}

\def\titlerunning{Non determinism through type isomorphism}
\def\authorrunning{A.~D\'iaz-Caro and G.~Dowek}

\begin{document}
\maketitle

\begin{abstract}
We define an equivalence relation on propositions and a proof system
where equivalent propositions have the same proofs.  The system
obtained this way resembles several known non-deterministic and
algebraic lambda-calculi.
\end{abstract}

\section{Introduction}
Several non-deterministic extensions to the $\lambda$-calculus have been proposed,
e.g.~\cite{BoudolIC94,BucciarelliEhrhardManzonettoAPAL12,deLiguoroPipernoIC95,DezaniciancaglinideliguoroPipernoTCS96,DezaniciancagliniDeliguoroPipernoSIAM98,PaganiRonchidellaroccaFI10}. In these approaches, the parallel composition (sometimes called the
{\em must-convergent} parallel composition)
is such that if $\ve r$ and $\ve s$ are two $\lambda$-terms, the term $\ve r+\ve s$ (also written $\ve r\parallel\ve s$) represents the computation that runs either $\ve r$ or $\ve s$ non-deterministically.
It is common to consider in these approaches the associativity and commutativity of the operator $+$. Indeed the interpretation ``either $\ve r$ or $\ve s$ runs'' shall not prioritise any of them, and so ``either $\ve s$ or $\ve r$ runs'' must be represented by the same term. Moreover, $(\ve r+\ve s)\ve t$ can run either $\ve r\ve t$ or $\ve s\ve t$, which is the same expressed by $\ve r\ve t+\ve s\ve t$. Extra equivalences (or rewrite rules, depending on the presentation) are set up to account for such an interpretation, e.g.~$(\ve r+\ve s)\ve t\leftrightarrow\ve r\ve t+\ve s\ve t$. This right distributivity can alternatively be seen as the one of function sum: $(\ve f+\ve g)(x)$ is defined pointwise as $\ve f(x)+\ve g(x)$. This is the approach of the algebraic lambda-calculi~\cite{ArrighiDowekRTA08,VauxMSCS09}, two independently introduced algebraic extensions which resulted strongly related afterwards~\cite{AssafPerdrixDCM11,DiazcaroPerdrixTassonValironHOR10}. In these algebraic calculi, a scalar pondering each `choice' is considered in addition to the sum of terms.

Because of these equivalences between terms, it is natural to think that a typed version must allow some equivalences at the type level. Definitely, if $\ve r$ and $\ve s$ are typed with types $A$ and $B$ respectively, it is natural to expect that whatever connective tie these types in order to type $\ve r+\ve s$, it must be commutative and associative.

An independent stream of research is the study of isomorphisms between types for several languages (see \cite{DiCosmo95} for a reference).
For example, we know that the propositions $A \wedge B$ and $B \wedge A$ are equiprovable: one is provable if and only if the other is, but they do not have the same proofs. If $\ve r$ is a proof of $A$ and $\ve s$ is a proof of $B$, then $\langle\ve r,\ve s\rangle$ is a proof of $A\wedge B$ while $\langle\ve s,\ve r\rangle$ is a proof of $B\wedge A$. Despite that both proofs can be derived from the same hypotheses, they are not the same.
In this paper, we show how the non-determinism arises naturally in a classic context only by introducing some equivalences between types. These equivalences, nevertheless, will be chosen among valid, well-known isomorphisms.
In order to consider these isomorphic types as equivalent, we need to design a proof system such that they have the same proofs, or conversely, in order to consider these terms to be equivalent, we need to make these isomorphic types to be equivalent. Formally, two types $A$ and $B$ are isomorphic if there are two conversion functions $f$ of type $A\Rightarrow B$ and $g$ of type $B\Rightarrow A$, such that $g(f(x))=x$ for any $x$ of type $A$ and $f(g(y))=y$ for any $y$ of type $B$.
Hence, in this system the conversion functions $f$ and $g$ should become and identity function. In other words, we take the quotient of the set of propositions by the relation generated by the isomorphisms of types and define proofs for elements in this quotient. In System F with products, which correspond to the propositional logic with universal quantifier, conjunction and implication, the full list of isomorphisms is known~\cite{DiCosmo95}, and it is summarised in Figure~\ref{fig:iso}.
\begin{figure}[!h]
\begin{multicols}{2}
 \begin{enumerate}
  \item\label{iso:comm} $A\wedge B\equiv B\wedge A$
  \item\label{iso:asso} $A\wedge (B\wedge C)\equiv (A\wedge B)\wedge C$
  \item\label{iso:distrib} $A\Rightarrow(B\wedge C)\equiv (A\Rightarrow B)\wedge(A\Rightarrow C)$
  \item\label{iso:currying} $(A\wedge B)\Rightarrow C\equiv A\Rightarrow (B\Rightarrow C)$
  \item\label{iso:ordering} $A\Rightarrow (B\Rightarrow C)\equiv B\Rightarrow(A\Rightarrow C)$
  \item $\forall X.\forall Y.A\equiv \forall Y.\forall X.A$
  \item $\forall X.A\equiv \forall Y.A[Y/X]$
  \item $\forall X.(A\Rightarrow B)\equiv A\Rightarrow \forall X.B$ if $X\notin FV(A)$
  \item $\forall X.(A\wedge B)\equiv \forall X.A\wedge\forall X.B$
  \item $\forall X.(A\wedge B)\equiv \forall X.\forall Y.(A\wedge (B[Y/X]))$
 \end{enumerate}
\end{multicols}
 \caption{All the type isomorphisms in propositional logic with universal quantifier, non-idempotent conjunction and implication}
 \label{fig:iso}
\end{figure}

In this work, we consider only the three first isomorphisms of this list,
because they are those that arise naturally when studying non deterministic
processes. The impact of the others is left for future work.

\begin{wrapfigure}{r}{0.1\textwidth}\vspace{-0.7cm}
$$\prooftree A\qquad B
  \justifies A \wedge B
  \endprooftree$$
\end{wrapfigure}Usually, for the deduction rule on the right
if we call $\ve r$ the proof of $A$ and $\ve s$ that of $B$, we write $\ve r, \ve s$ or $\langle\ve r,\ve s\rangle$ the proof of $A \wedge B$. However if $A \wedge B$ and $B \wedge A$ are {\em the same} proposition, we get $\ve r, \ve s$ and $\ve s, \ve r$ to be the same term. Let us write ``$+$'' to the commutative comma\footnote{We could chose another symbol, however $+$ is the one used in most non-deterministic settings.} and set the rule
$$\prooftree{ \ve r:A\qquad \ve s:B}
  \justifies{ \ve r + \ve s:A \wedge B}
  \endprooftree\ .$$

In the same way, the associativity of $\wedge$ induces that of $+$.
Furthermore, the isomorphism (\ref{iso:distrib}) of Figure~\ref{fig:iso} induces the following equivalence on proofs. If $\ve r$ is a proof of $A \Rightarrow B$, $\ve s$ one of $A \Rightarrow C$, and $\ve t$ one of $A$ then $\ve r + \ve s$ is a proof of $A \Rightarrow (B \wedge C)$ and $(\ve r + \ve s)\ve t$ is a proof of $B \wedge C$. This proof is the same as $\ve r\ve t + \ve s\ve t$.
Summarising, from the equivalences between types we obtained a commutative and associative $+$, which is such that the application right-distributes over it.

Several non-classical type systems have been already proposed for the non-deterministic and algebraic calculi, e.g.~\cite{ArrighiDiazcaroLMCS12,ArrighiDiazcaroValironDCM11,DiazcaroPetitWoLLIC12}. In these systems there is already an equivalence relation on propositions such that if $A\equiv B$ and $A$ types a term, then also $B$ types it. Such equivalence is reminiscent of type theory~\cite{CoquandHuetIC88,MartinLof84} and deduction modulo~\cite{DowekHardinKirchnerJAR03,DowekWernerJSL98}. But here we go further, introducing an equivalence relation that equates types built with different connectives such as $A\Rightarrow(B\wedge C)$ and $(A\Rightarrow B)\wedge(A\Rightarrow C)$, which is not possible there.
Moreover, there is no elimination rule for conjunction in~\cite{ArrighiDiazcaroLMCS12,ArrighiDiazcaroValironDCM11,DiazcaroPetitWoLLIC12}. Indeed, having commutativity and associativity properties in both, the sums of terms and the conjunctions of propositions, leads to uncertainty on how to eliminate them. A rule like
``$\ve r:A\wedge B$ implies $ \pi_1(\ve r):A$'',
would not be consistent.
If $A$ and $B$ are two arbitrary types, $\ve s$ a term of type $A$ and $\ve t$ a term of type $B$, then $\ve s+\ve t$ has both types $A\wedge B$ and $B \wedge A$, thus $\pi_1(\ve s +\ve t)$ would have both type $A$ and type $B$.
Hence, a naive rule would lead to inconsistency. The projection would project a random term of any of the types of its arguments, so not being a trustfully valid proof for any proposition.

The approach we follow here is to consider explicitly typed terms (Church style), and hence make the projection to depend on the type:
if $\ve r:A\wedge B$ then $\pi_A(\ve r):A$.
This way, we recover consistency of the proof system. This new form of projection entails allowing some non-determinism directly in the rewrite system. Indeed, if $\ve r$ and $\ve s$ have the same type $A$,
$\pi_A(\ve r +\ve s)$ both reduces to $\ve r$ and to $\ve s$.
A priori this does not entail any problem; any of them is a valid proof of the same proposition $A$. This approach can be summarised by the slogan {\em ``the subject reduction property is more important than the uniqueness of results''}~\cite{DowekJiangIC11}. Therefore the projection turns the non-deterministic choice explicit.

We formalise all of the previously discussed concepts in Section~\ref{sec:calculus}, where we present the calculus \OurCalculus, and provide some examples. Section~\ref{sec:SR} The next section is devoted to prove that our system enjoys the subject reduction property. In Section~\ref{sec:prob} we discuss the relation of this setting with respect to the algebraic approach. Finally, Section~\ref{sec:conclusion} concludes the paper with suggestions for future research.

\section{The calculus}\label{sec:calculus}
\subsection{Definitions}
In this section we present the calculus \OurCalculus, an explicitly typed lambda-calculus extended with a $+$ operator as discussed in the introduction.
We consider the following grammar of types
$$A,B,C,\dots\ ::=\quad X~|~A\Rightarrow B~|~A\wedge B~|~\forall X.A\enspace,$$
where the isomorphisms (\ref{iso:comm}), (\ref{iso:asso}) and (\ref{iso:distrib}) from Figure~\ref{fig:iso} are made explicit by an equivalence relation between types
$$A\wedge B~\equiv~ B\wedge A\quad ,\quad
(A\wedge B)\wedge C~\equiv~ A\wedge(B\wedge C)\quad \mbox{and}\quad
A\Rightarrow (B\wedge C)~\equiv~ (A\Rightarrow B)\wedge(A\Rightarrow C)\enspace.$$

The set of terms $\Lambda$ is defined inductively by the grammar
$$\ve r,\ve s,\ve t\ ::=\quad x^A~|~\lambda x^A.\ve r~|~\ve r\ve s~|~\ve r+\ve s~|~\pi_{A}(\ve r)~|~\Lambda X.\ve r~|~\ve r\{A\}\enspace.$$

All our variable occurrences are explicitly typed, but we usually omit the superscript indicating the type of variables when it is clear from the context. For example we write $\lambda x^A.x$ instead of $\lambda x^A.x^A$.
The {\em $\alpha$-conversion} and the sets $FV(\ve r)$ of {\em free variables of $\ve r$} and $FV(A)$ of {\em free variables of $A$} are defined as usual in the $\lambda$-calculus (cf.~\cite[\S2.1]{Barendregt84}). For example $FV(x^Ay^B)=\{x^A,y^B\}$.
The same variable, with different types, is treated as a different variable. For example, the term $\lambda x^A.x^B:A\Rightarrow B$ is typable in our system, and it is the constant function $x^B$, since $x^B$ is free in the term $\lambda x^A.x^B$.
We say that a term $\ve r$ is {\em closed} whenever $FV(\ve r)=\emptyset$. Given two terms $\ve r$ and $\ve s$ we denote by $\ve r[\ve s/x]$ the term obtained by simultaneously substituting the term $\ve s$ for all the free occurrences of $x$ in $\ve r$, subject to the usual proviso about renaming bound variables in $\ve r$ to avoid capture of the variables free in $\ve s$. Analogously $A[B/X]$ denotes the substitution of the type $B$ for all the free occurrences of $X$ in $A$, and $\ve r[B/X]$ the substitution in $\ve r$. For example, $(x^A)[B/Y]=x^{(A[B/Y])}$, $(\lambda x^A.\ve r)[B/X]=\lambda x^{A[B/X]}.\ve r[B/X]$ and $(\pi_A(\ve r))[B/X]=\pi_{A[B/X]}(\ve r[B/X])$. Simultaneous substitutions are defined in the same way. Finally, terms and types are considered up to $\alpha$-conversion.

Each term of the language has a main type associated, which can be obtained from the type annotations, and other types induced by the type equivalences.
The type system for \OurCalculus\ is given in Figure~\ref{fig:typeSys}.
If $FV(\ve r)=\{x_1^{A_1},\dots,x_n^{A_n}\}$, we write $\Gamma(\ve r)=\{A_1,\dots,A_n\}$. $FV(\{A_1,\dots,A_n\})$ is defined by $\bigcup_{i=1}^nFV(A_i)$.
{\em Typing judgements} are of the form $\ve r:A$. A term $\ve r$ is {\em typable} if there exists a type $A$ such that $\ve r:A$.

\begin{figure}[!ht]
\centering\vspace{0.1cm}
\hspace{0.25cm}
  \prooftree
  \justifies x^A:A
  \using{ax}
  \endprooftree
\hfill
  \cond{A\equiv B}\prooftree\ve r: A
  \justifies\ve r: B
  \using\equiv
  \endprooftree
\hfill
  \prooftree\ve r:B
  \justifies \lambda x^A.\ve r:A\Rightarrow B
  \using\Rightarrow_I
  \endprooftree
\hfill
  \prooftree\ve r:A\Rightarrow B\quad \ve s:A
  \justifies\ve r\ve s:B
  \using\Rightarrow_E
  \endprooftree\hspace{0.10cm}
\vspace{0.5cm}

\hspace{0.25cm}
  \prooftree\ve r:A\quad \ve s:B
  \justifies\ve r+\ve s:A\wedge B
  \using\wedge_I
  \endprooftree
\hfill
  \prooftree\ve r:A\wedge B
  \justifies\pi_A(\ve r):A
  \using\wedge_E
  \endprooftree
\hfill
  \cond{X\notin FV(\Gamma(\ve r))}\prooftree\ve r: A
  \justifies\Lambda X.\ve r:\forall X.A
  \using\forall_I
  \endprooftree
\hfill
  \prooftree\ve r:\forall X.A
  \justifies\ve r\{B\}:A[B/X]
  \using\forall_E
  \endprooftree\hspace{0.10cm}

\vspace{0.1cm}
  \caption{The type system for \OurCalculus}
  \label{fig:typeSys}
\end{figure}

Lemma~\ref{lem:unicity} states that the typing modulo equivalences is unique.
\begin{lemma}\label{lem:unicity}
 If $\ve r:A$ and $\ve r:B$, then $A\equiv B$.
\end{lemma}
\begin{proof}
  Without rule $\equiv$, the type system is syntax directed. The only rule able to modify the type of a term without changing it is $\equiv$.
\end{proof}

The operational semantics of the calculus is given in Figure~\ref{fig:opSem}, where there are two distinct relations between terms: $\re$ and a symmetric relation $\eq$.
We write $\eq^*$ and $\re^*$ for the transitive and reflexive closures of $\eq$ and $\re$ respectively.
In particular, notice that $\eq^*$ is an equivalence relation.

\begin{figure}[!h]\centering
    \emph{Symmetric relation:}\\
    \hspace{1cm}$\ve r+\ve s\eq\ve s+\ve r$,\hfill
	$(\ve r+\ve s)+\ve t\eq\ve r+(\ve s+\ve t)$,\hfill
    $(\ve r+\ve s)\ve t\eq\ve r\ve t+\ve s\ve t$,\hspace{1cm}

	\hspace{1cm}$\lambda x^A.(\ve r+\ve s)\eq\lambda x^A.\ve r+\lambda x^A.\ve s$,\hfill
	If $\ve r:A\Rightarrow (B\wedge C)$, then $\pi_{A\Rightarrow B}(\ve r)\ve s\eq\pi_B(\ve r\ve s)$.\hspace{1cm}
\vspace{0.3cm}

	\emph{Reductions:}\\
    \hspace{1cm}$(\lambda x^A.\ve r)~\ve s\re\ve r[\ve s/x]$,\hfill
	$(\Lambda X.\ve r)\{A\}\re\ve r[A/X]$,\hfill
    If $\ve r:A$, then $\pi_A(\ve r+\ve s)\re\ve r$.\hspace{1cm}

	\caption{Operational semantics of \OurCalculus}
	\label{fig:opSem}
\end{figure}

\subsection{Examples}
\begin{example} We have
 $\lambda x^{A\wedge B}.x:(A\wedge B)\Rightarrow (A\wedge B)$ and so by rule $\equiv$,
 $\lambda x^{A\wedge B}.x:((A\wedge B)\Rightarrow A)\wedge((A\wedge B)\Rightarrow B)$, from which we can obtain $\pi_{(A\wedge B)\Rightarrow A}(\lambda x^{A\wedge B}.x):(A\wedge B)\Rightarrow A$.
 Let $\ve r:A\wedge B$, then $\pi_{(A\wedge B)\Rightarrow A}(\lambda x^{A\wedge B}.x)\ve r:A$,
and notice that
 $\pi_{(A\wedge B)\Rightarrow A}(\lambda x^{A\wedge B}.x)\ve r~\eq~\pi_A((\lambda x^{A\wedge B}.x)\ve r)~\re~\pi_A(\ve r)$.
\end{example}

\begin{example} Let $\tf=\lambda x^A.\lambda y^B.(x+y)$. It is easy to check that $ \tf:A\Rightarrow B\Rightarrow (A\wedge B)$, and by rule $\equiv$ it also has the type $(A\Rightarrow B\Rightarrow A)\wedge (A\Rightarrow B\Rightarrow B)$. Therefore, $\pi_{A\Rightarrow B\Rightarrow A}(\tf):A\Rightarrow B\Rightarrow A$ is well typed. In addition, if $\ve r:A$ and $\ve s:B$, we have $\pi_{A\Rightarrow B\Rightarrow A}(\tf)\ve r\ve s:A$.

Notice that
$\pi_{A\Rightarrow B\Rightarrow A}(\tf)\ve r\ve s~\eq~
\pi_{B\Rightarrow A}(\tf\ve r)\ve s~\eq~
\pi_{A}(\tf\ve r\ve s)~\re~
\pi_A((\lambda y^B.(\ve r+y))\ve s)~\re~
\pi_A(\ve r+\ve s)\re\ve r$,
which is coherent with such typing.
\end{example}

\begin{example}
 Let $\true=\lambda x^A.\lambda y^B.x$ and $\false=\lambda x^A.\lambda y^B.y$. Then $ \true+\false:(A\Rightarrow B\Rightarrow A)\wedge (A\Rightarrow B\Rightarrow B)$, hence $\pi_{(A\Rightarrow B\Rightarrow A)\wedge(A\Rightarrow B\Rightarrow B)}(\true+\false+\tf)$ reduces non-deterministically either to $\true+\false$ or to $\tf$. Moreover, notice that $\true+\false\eq^* \tf$, hence in this very particular case, the non-deterministic choice does not play any role.
\end{example}

\section{Subject reduction}\label{sec:SR}
In this section we prove that the set of types assigned to a term is invariant under $\eq$ and $\re$.
In other words, Theorem~\ref{thm:SR} states that if $\ve r$ is a proof of $A$, any reduction fired from $\ve r$ will still be a proof of $A$.

The substitution lemma below will be the key ingredient in the proof of subject reduction.
It ensures that when substituting types for type variables or terms for term variables,
in an adequate manner, the typing judgements remain valid.

\begin{lemma}[Substitution]\label{lem:substitution}~
 If $\ve r:B$ and $\ve s:A$, then $\ve r[\ve s/x^A]:B$. Also,
 If $\ve r:A$, then $\ve r[B/X]:A[B/X]$.
\end{lemma}
\begin{proof} By induction over $\ve r$ for the first result and over the type derivation for the second.
\end{proof}

\noindent Now we can prove the subject reduction property, ensuring that the typing is preserved during reduction.

\begin{theorem}[Subject reduction]\label{thm:SR}
 If $\ve r\to\ve s$ and $\ve r:A$, then $\ve s:A$\quad (where $\to$ is either $\re$ or $\eq$).
\end{theorem}
\begin{proof}
 By induction over the reduction relation. We give only two interesting cases.

\noindent\textbf{Rule  $\pi_{A\Rightarrow B}(\ve r)\ve s\eq\pi_B(\ve r\ve s)$, with $\ve r:A\Rightarrow(B\wedge C)$}. Let $\pi_{A\Rightarrow B}(\ve r)\ve s:D$, then $\pi_{A\Rightarrow B}(\ve r):E\Rightarrow D$ and $\ve s:E$. But then $E\equiv A$ and $D\equiv B$, because clearly, the main type for $\pi_{A\Rightarrow B}(\cdot)$ is $A\Rightarrow B$, so $\ve r:(A\Rightarrow B)\wedge F$, however since $\ve r:A\Rightarrow(B\wedge C)$, we have $F\equiv A\Rightarrow C$. So, by rule $\Rightarrow_E$, $\ve r\ve s:B\wedge C$. We conclude by rule $\wedge_E$.
  For the inverse direction, let $\pi_B(\ve r\ve s):D$. Then $D\equiv B$ and $\ve r\ve s:B\wedge E$, so $\ve r:F\Rightarrow(B\wedge E)$ and $\ve s:F$. Hence, since $\ve r:A\Rightarrow(B\wedge C)$, by Lemma~\ref{lem:unicity}, we have $F\equiv A$ and $E\equiv C$, so $\pi_{A\Rightarrow B}(\ve r):A\Rightarrow B$, from which, we conclude $\pi_{A\Rightarrow B}(\ve r)\ve s:B$. We conclude by rule $\equiv$.

\noindent\textbf{Rule $(\lambda x^A.\ve r)~\ve s\re\ve r\subst{\ve s}{x}$.} Let $(\lambda x^A.\ve r)\ve s:B$, then $\lambda x^A.\ve r:C\Rightarrow D$ and $\ve s:C$, with $D\equiv B$. Then $\ve r:E$, with $A\Rightarrow E\equiv C\Rightarrow D$. Notice that, since $A\Rightarrow E\equiv C\Rightarrow D$, it must be $A\equiv C$ and $E\equiv D$. Then, by rule $\equiv$, $\ve s:A$, and so, by Lemma~\ref{lem:substitution}, $\ve r[\ve s/x^A]:E$, and since $E\equiv D\equiv B$, by rule $\equiv$, we obtain $\ve r[\ve s/x^A]:B$.
%
%
%
%
  \qedhere
\end{proof}

\section{From non-determinism to probabilities}\label{sec:prob}
In~\cite{ArrighiDowekRTA08} and \cite{VauxMSCS09} two algebraic extensions of the untyped lambda-calculus are introduced, which we call $\lambda_{\textrm{lin}}$ and $\lambda_{\textrm{alg}}$ respectively. In these settings, not only the $+$ operator is present, but also a scalar pondering each choice. Hence, if $\ve r$ and $\ve s$ are two possible terms, so is the linear combination of them $\alpha.\ve r+\beta.\ve s$, with $\alpha,\beta$ some kind of scalars (taken from a generic ring in $\lambda_{\textrm{lin}}$ or from $\mathbb{R}^{\geq 0}$ in $\lambda_{\textrm{alg}}$). Both these calculi identify the term $(\ve r+\ve s)\ve t$ with $\ve r\ve t+\ve s\ve t$, either with a rewrite system or an equality, and $+$ is associative and commutative. Also, the scalars interact with the $+$, e.g.~$\ve r+\ve r\leftrightarrow 2.\ve r$. By restricting the scalars to positive real numbers,
or even to natural numbers, one possible interpretation is that the scalars give the probability of following one possible path (after `normalising' the scalars, i.e.~dividing over the total amount in order to sum up to $1$).
In this way, the term $2.\ve r+\ve s$ is twice more likely to run $\ve r$ than $\ve s$.

Indeed, in \cite[\S6]{ArrighiDiazcaroLMCS12} the type system $\mathsf{B}$ for $\lambda_{\textrm{lin}}$ is proposed, which can decide whether a superposition is a probability distribution (i.e.~it can check that the sum of terms is up to $1$). Such a system includes scalars at the type level, reflecting those in the terms, so $\alpha.\ve r$ has type $\alpha.A$ whenever $\ve r$ has type $A$. This provides a powerful tool to account for the scalars within the terms, however it entails a `non-classical' extension of System F with scalars pondering the types. In such a formalism, there is no possibility to tie terms with different types: if $\ve r$ and $\ve s$ have both type $A$, then $\alpha.\ve r+\beta.\ve s$ have type $(\alpha+\beta).A$, however if the types of $\ve r$ and $\ve s$ differ, the previous term cannot be typed.
That weakness is solved in~\cite{ArrighiDiazcaroValironDCM11}, where a more powerful system is introduced, with a type system also allowing for linear combination of types, just like for terms. In both these systems, while powerful, it is hard to establish a connection with a well-known logic. That is precisely the goal of~\cite{BuirasDiazcaroJaskelioffLSFA11}, where a more `classic' system is developed, with no scalars at the type level. However it carries some costs: first, it is only meant for positive real scalars (which anyway is enough for a `probabilistic' interpretation), and more importantly, the type system gives just an approximation, an upper bound, of the scalars in the terms.

We could envisage extending \OurCalculus\ with a more thorough projection where $\pi_A(\alpha.\ve r+\beta.\ve s)$ would output
\begin{wrapfigure}{r}{0.37\textwidth}
  \begin{center}\vspace{-0.2in}
$\xymatrix@C=5ex@R=4ex{
   & \pi_A(\ve r+\pi_A(\ve s+\ve t)+\ve t) \ar[ddl] \ar[ddr]\ar[d] & \\
    & \pi_A(\ve s+\ve t)\ar[d] \ar[dr] & \\
   \ve r & \ve s & \ve t
  }$\end{center}
\end{wrapfigure}
either $\ve r$, with probability $\alpha$, or $\ve s$ with probability $\beta$. However, even when the scalars are not explicitly written, the probabilities are present. The following example is clarifying.

Let $\ve r:A$, $\ve s:A$ and $\ve t:A$. Then, the reductions depicted in the diagram at right are possible. If we consider $\pi_A$ making an equiprobable choice instead of a non-deterministic one, it is clear that $\ve t$ have more probability to be reached, followed by $\ve r$, and the less likely is $\ve s$.

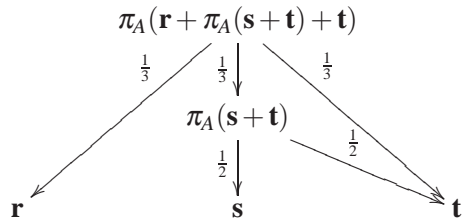
\begin{wrapfigure}{L}{0.39\textwidth}
  \begin{center}\vspace{-0.2in}
$\xymatrix@C=6ex@R=4ex{
   & \pi_A(\ve r+\pi_A(\ve s+\ve t)+\ve t) \ar[ddl]_(.36){\frac{1}{3}} \ar[ddr]^(.35){\frac{1}{3}} \ar[d]_{\frac{1}{3}} & \\
    & \pi_A(\ve s+\ve t)\ar[d]_{\frac{1}{2}} \ar[dr]^{\frac{1}{2}} & \\
   \ve r & \ve s & \ve t
  }$
    \end{center}
\end{wrapfigure}

Indeed, we can calculate the global probability of reaching each possibility by labelling the reductions with its local probability as shown in the diagram at left, from where just by summing up the labels reaching a term, and multiplying those in the same path, we can easily check that the term $\ve r$ has probability $\frac{1}{3}$ of being reached, the term $\ve s$ probability $\frac{1}{6}$ and the term $\ve t$ probability $\frac{1}{2}$.
Hence, this term would be expressed with
scalars as $\frac{1}{3}\ve r+\frac{1}{6}\ve s+\frac{1}{2}\ve t$ according to the
previously discussed interpretation.
Therefore, \OurCalculus\ could be seen as a sort of algebraic calculus, with implicit scalars taken from $\mathbb{Q}^{[0,1]}$, typed with a standard type system. These ideas will be fully developed in a future research.

\section{Conclusions and future work}\label{sec:conclusion}
\subsection{Conclusions}
In this paper we have introduced \OurCalculus, a proof system for second order propositional logic with an associative and commutative conjunction, and implication. In this system, isomorphic propositions get the same proofs. At this first step we only consider three isomorphisms, namely commutativity and associativity of the conjunction, and distributivity of implication with respect to conjunction. We use the symbol $+$ to put together the proofs of different propositions, so $\ve r+\ve s$ becomes a proof of $A\wedge B$, if $\ve r$ is a proof of $A$ and $\ve s$ a proof of $B$. Such a symbol is commutative and associative, and application is right-distributive with respect to it, to account for the isomorphisms of propositions.

This construction entails a non-deterministic projection where if a proposition has two possible proofs, the projection of its conjunction can output any of them. For example, if $\ve r$ and $\ve s$ are two possible proofs of $A$, then $\pi_A(\ve r+\ve s)$ will output either $\ve r$ or $\ve s$.

In several works (cf.~\cite[\S3.4]{ManzonettoPhDThesis} for a reference), the non-determinism is modelled by two operators. The first  is normally written $+$, and instead of distributing over application, it actually makes the non-deterministic choice. Hence $(\ve r+\ve s)\ve t$ reduces either to $\ve r\ve t$ or to $\ve s\ve t$~\cite{deLiguoroPipernoIC95}. The second one, denoted by $\parallel$, does not make the choice, and therefore $(\ve r\parallel\ve s)\ve t$ reduces to $\ve r\ve t\parallel\ve s\ve t$~\cite{DezaniciancagliniDeliguoroPipernoSIAM98}. One way to interpret these operators is that the first one is a non-deterministic one, while the second is the parallel composition. Another common interpretation is that $+$ is a {\it may-convergent} non-deterministic operator, where type systems ensure that at least one branch converges, while $\parallel$ is a {\it must-convergent} non-deterministic operator, where both branches are meant to converge~\cite{BucciarelliEhrhardManzonettoAPAL12,DiazcaroManzonettoPaganiLFCS13}. In our setting, the $+$ operator in \OurCalculus\ behaves like $\parallel$, and an extra operator ($\pi_A$) induces the non-deterministic choice.
The main point is that this construction arose naturally just by considering some of the isomorphisms between types as an equivalence relation.
In order to ensure that our system is must-convergent, we shall prove its strong normalisation, which is left for future research.

\subsection{Open questions and future research}
As mentioned in Section~\ref{sec:prob}, the calculus \OurCalculus\ has implicit scalars on it, which can convert this non-deterministic setting into a probabilistic one. The original motivation behind $\lambda_{\textrm{lin}}$~\cite{ArrighiDowekRTA08} and its {\em vectorial} type system~\cite{ArrighiDiazcaroValironDCM11} was to encode quantum computing on it. A projection depending on scalars could lead to a measurement operator in a future design---after other questions like deciding orthogonality~\cite{ValironQPL10} have been addressed in that setting. This is a promising future direction we are willing to take.

In order to follow such direction, a first step is to move to a call-by-value calculus, where $\ve r(\ve s+\ve t)\eq\ve r\ve s+\ve r\ve t$ (because a non-deterministic choice yet to make, is not considered to be a value). The reason to move to call-by-value is explained with the following example. Consider for instance the term $\ve{$\delta$}=\lambda x.xx$ applied to a sum $\ve r+\ve s$. In call-by-name it reduces to $(\ve r+\ve s)(\ve r+\ve s)$ while in a call-by-value strategy  ($\lambda_{\textrm{lin}}$) the same term reduces to $\ve{$\delta$}\ve r+\ve{$\delta$}\ve s$ first, and then to $\ve r\ve r+\ve s\ve s$. If seeking for a quantum interpretation, reducing $\ve{$\delta$}(\ve r+\ve s)$ into $(\ve r+\ve s)(\ve r+\ve s)$ is considered as the forbidden quantum operation of ``cloning''~\cite{WoottersZurekNATURE82}, while the alternative reduction to $\ve r\ve r+\ve s\ve s$ is seen as a ``copy'', or CNOT, a fundamental quantum operation~\cite{MonroeMeekhofKingItanoWinelandPRL95}.

In order to account for such an equivalence, $\ve r(\ve s+\ve t)\eq\ve r\ve s+\ve r\ve s$, we would need an equivalence at the type level such as
$(A\wedge B)\Rightarrow C\equiv (A\Rightarrow C)\wedge(B\Rightarrow C)$,
however it is clearly false. A workaround which have been used already in the vectorial type system~\cite{ArrighiDiazcaroValironDCM11} is to use the polymorphism instead of an equivalence. If $\ve r$ have type $\forall X.X\Rightarrow C_X$, then we can specialise $X$ to the needed argument. Indeed,
$\forall X.X\Rightarrow C_X$ entails both $A\Rightarrow C_A$ and $B\Rightarrow C_B$, which can latter be tied by a conjunction.
\medskip

Another prominent future work is to determine what is needed for the remaining isomorphisms (cf.~Figure~\ref{fig:iso}).
In a work by Garrigue and A\"it-Kaci~\cite{GarrigueAitkaciPOPL94}, the isomorphism $A\wedge B\equiv B\wedge A$ has been indirectly treated by combining it with currying: $(A\wedge B)\Rightarrow C\equiv A\Rightarrow B\Rightarrow C$ (cf.~isomorphism~(\ref{iso:currying}) of Figure~\ref{fig:iso}), from which it can be deduced the isomorphism $A\Rightarrow (B\Rightarrow C)\equiv B\Rightarrow (A\Rightarrow C)$ (cf.~isomorphism~(\ref{iso:ordering}) of Figure~\ref{fig:iso}). Their proposal is the selective $\lambda$-calculus, a calculus including labellings to identify which argument is being used at each time. Moreover, by considering the Church encoding of pairs, isomorphism (\ref{iso:ordering}) implies isomorphism (\ref{iso:comm}) (commutativity of $\wedge$). However their proposal is completely different to ours, and the non-determinism cannot be inferred from the selective $\lambda$-calculus.

\paragraph{Acknowledgements.} We would like to thank Fr\'ed\'eric Blanqui, Michele Pagani and Giulio Manzonetto for enlightening discussions.

\bibliographystyle{eptcs}
\bibliography{biblio}

\end{document}